\documentclass[11pt,onecolumn,draftcls]{IEEEtran}
\usepackage{multicol}
\usepackage{amsmath,amssymb,cite,multirow,subfigure, amsthm}
\usepackage{graphicx}
\usepackage{url}
\usepackage{algorithm}
\usepackage{algorithmic}
\usepackage{makecell}
\usepackage{flushend}

\newtheorem{example}{Example}
\newtheorem{lemma}{Lemma}

\newcommand{\df}{\stackrel{\mbox{\scriptsize def}}{=}}

\renewcommand{\Pi}{P_{\mbox{\tiny{I}}}}

\begin{document}
\title{General Linearized Polynomial Interpolation\\and Its Applications}

\author{Hongmei Xie, Zhiyuan Yan,~\IEEEmembership{Senior Member, IEEE}, and Bruce W. Suter,~\IEEEmembership{Senior Member, IEEE}
\thanks{Hongmei Xie and Zhiyuan Yan are with the Department of Electrical and Computer
Engineering, Lehigh University, Bethlehem, PA 18015, USA (E-mails:
{\tt hox209,yan@lehigh.edu}). Bruce W. Suter is with Air Force Research Laboratory, Rome, NY 13441, USA
(E-mail: {\tt bruce.suter@rl.af.mil}).}}

\maketitle

\thispagestyle{empty}
\begin{abstract}
In this paper, we first propose a general interpolation algorithm in a free module of a \textbf{linearized} polynomial ring, and then apply this algorithm to decode several important families of codes, Gabidulin codes, KK codes and MV codes. Our decoding algorithm for Gabidulin codes is different from the polynomial reconstruction algorithm by Loidreau. When applied to decode KK codes, our interpolation algorithm is equivalent to the Sudan-style list-1 decoding algorithm proposed by K\"{o}tter and Kschischang for KK codes. The general interpolation approach is also capable of solving the interpolation problem for the list decoding of MV codes proposed by Mahdavifar and Vardy, and has a lower complexity than solving linear equations.
\end{abstract}

\section{Introduction}\label{sec:intro}
Given a set of points, polynomial interpolation finds one or more polynomials that go through these points. Since error correcting codes are often defined through polynomials, polynomial interpolation is instrumental in decoding such error control codes. For instance, Reed-Solomon (RS) codes can be defined using evaluation of polynomials~\cite{RothBookRS}, and bivariate polynomial interpolation has been used in RS decoders. In particular, the K\"{o}tter interpolation~\cite{kotter_it96} implements the interpolation step of the Guruswami-Sudan algorithm~\cite{GS99} for RS codes with low complexity. Also, the Welch-Berlekamp key equation can be viewed as a rational interpolation problem, and the Welch-Berlekamp algorithm (WBA) solves this problem~\cite{WelchBerlekamp86}.

Polynomial interpolation was extended by Wang \emph{et al.}~\cite{WangMcEliece05} to a general interpolation problem in a free module that is defined over a polynomial ring over some finite field $F$ and admits an ordering. Since the free module is also a vector space over $F$, one can define linear functionals on the free module. Given any set of linear functionals, the general interpolation problem is to find a minimum element in the intersection of the kernels of the linear functionals. Wang \emph{et al.} proposed a general interpolation algorithm, and showed that the K\"{o}tter interpolation and the WBA are both special cases of this general interpolation algorithm \cite{WangMcEliece05}.

Recently, error control codes defined using evaluation of \textbf{linearized polynomials} have attracted growing attention, such as Gabidulin codes~\cite{gabidulin_pit0185} and a family of subspace codes proposed by K\"{o}tter and Kschischang \cite{kotter_it08}, referred to as KK codes. While both Gabidulin and KK codes are important to error control in random linear network coding (see, for example, \cite{kotter_it08,silva_it08,silva_it09}), Gabidulin codes are also considered for potential applications in wireless communications \cite{lusina_it03}, public-key cryptosystems \cite{gabidulin_lncs91}, and storage systems
\cite{gabidulin_pit0285, RothBookRS}. A decoding algorithm of Gabidulin codes through linearized polynomial reconstruction was proposed by Loidreau~\cite{loidreau_wcc05}, and K\"{o}tter and Kschischang proposed a Sudan-style list-1 decoding algorithm for KK codes based on bivariate linearized polynomial interpolation \cite{kotter_it08}. Following similar list decoding idea for RS codes by Guruswami and Sudan~\cite{GS99}, Mahdavifar and Vardy considered list decoding of KK codes in~\cite{mahdavifar_isit10}\cite{mahdavifar_it10}, where the construction of KK codes were modified accordingly. Codes in ~\cite{mahdavifar_isit10}\cite{mahdavifar_it10} are similar to but different from KK codes, and we call the new class of subspaces codes MV codes.

Parallel to the work of Wang \emph{et al}.~\cite{WangMcEliece05}, we investigate the general interpolation problem in a free module of a \textbf{linearized polynomial ring}. The main contributions of this paper are listed as follows.
\begin{itemize}
\item We propose a general interpolation algorithm in a free module of a linearized polynomial ring, and show that our interpolation algorithm has a polynomial time complexity.
\item We apply our interpolation algorithm to decode Gabidulin codes. The resulted decoding algorithm resembles Loidreau's decoding algorithm \makeatletter
    \renewcommand\@cite [1]{#1}
    \makeatother (cf. [\cite{loidreau_wcc05}, Table~1]), \makeatletter
    \renewcommand\@cite [1] {[#1]}
    \makeatother and both algorithms have quadratic complexity, but the two differ in several key aspects.
\item Our general interpolation approach is also used to decode KK codes. In fact, in this case, our algorithm is equivalent to the Sudan-style list-1 decoding algorithm in~\cite{kotter_it08}. That is, the Sudan-style list-1 decoding algorithm is a special case of our general interpolation algorithm, when some operations and parameters are specified.
\item Finally, we use our general interpolation algorithm to obtain the multivariate polynomial for the list decoding of MV codes in~\cite{mahdavifar_isit10}. To the best of our knowledge, there is no other efficient algorithm to accomplish the task. We also show that our algorithm has lower complexity than solving linear equations.
\end{itemize}

The rest of the paper is organized as follows. Section~\ref{sec:preliminaries} reviews the general interpolation over free modules of polynomial rings, and then introduces Gabidulin codes, KK codes and MV codes, as well as their respective decoding algorithms. In Section~\ref{sec:generalinterpo}, we propose our general interpolation algorithm over a free module of a linearized polynomial ring, and analyze its computational complexity. We apply our general interpolation algorithm to decode Gabidulin codes as well as KK codes and MV codes in Sections~\ref{sec: GeneralGb}, \ref{sec: generalKK}, and \ref{sec: generalMV}, respectively. Concluding remarks are provided in Section~\ref{sec: conclusion}.

\section{Preliminaries} \label{sec:preliminaries}
\subsection{General Polynomial Interpolation over Polynomials Ring}
Motivated by the K\"{o}tter interpolation, Wang \emph{et al}.~\cite{WangMcEliece05} consider a general interpolation problem. Let $F[x]$ be the ring of all the polynomials over some finite field $F$. A free $F[x]$-module $V$ is an $F[x]$-module with a basis. Suppose $V$ is also a vector space over $F$ with a basis $M$, then we can define a set of $C$ linear functionals $D_i$'s from $V$ to $F$, with corresponding kernels $K_i$'s, where $i = 1,2, \ldots, C$. If there is a total ordering on $M$, $V$ admits an ordering. That is, for a subset of $V$ we can find an element with the smallest order, and the element is a \emph{minimum} in this subset. The general interpolation algorithm in~\cite{WangMcEliece05} finds a minimum in $K_1\cap K_2\cap \cdots \cap K_{C}$.

\subsection{Linearized Polynomial Ring} \label{sec: L[x]}
Suppose GF$(q^m)$ is an extension field of GF$(q)$, where $q$ is a prime power and $m$ is a positive integer. A polynomial of the form
\begin{equation}
l(x) = \sum_{i = 0}^n a_ix^{q^i}
\end{equation}
with coefficients $a_i \in$ GF$(q^m)$ is called a \emph{linearized polynomial} over GF$(q^m)$. We assume $q$ is fixed, and denote $x^{q^i}$ as $x^{[i]}$ in this paper. For a linearized polynomial $l(x) = \sum_{i = 0}^n a_ix^{[i]}$ over GF$(q^m)$, its $q$-degree, denoted as deg$_q(l(x))$, is given by $\max\limits_{a_i \ne 0}\left\{i\right\}$.

Linearized polynomials are so named because for a linearized polynomial $l(x)$ over GF$(q^m)$, $\beta_1$ and $\beta_2$ in an extension field $\mathbb{K}$ of GF$(q^m)$, and $\lambda_1,\lambda_2 \in$ GF$(q)$, we have $l(\lambda_1\beta_1 + \lambda_2\beta_2)  = \lambda_1l(\beta_1) + \lambda_2 l(\beta_2)$. In other words, $l(x)$ can be treated as a linear mapping from $\beta \in \mathbb{K}$ to $l(\beta) \in \mathbb{K}$ with respect to GF$(q)$~\cite{lidl_book83}. Given two linearized polynomials $l_1(x)$ and  $l_2(x)$ over GF$(q^m)$, their GF$(q^m)$-linear combination $\alpha_1l_1(x) + \alpha_2l_2(x)$ with $\alpha_1, \alpha_2 \in$ GF$(q^m)$, is also a linearized polynomial over GF$(q^m)$. We define the multiplication between $l_1(x)$ and  $l_2(x)$ as $l_1(x) \otimes l_2(x) \df l_1(l_2(x))$, and $l(x) = l_1(x) \otimes l_2(x)$ is also a linearized polynomial over GF$(q^m)$. Note that generally $l_1(x) \otimes l_2(x)$ does not necessarily equal $ l_2(x) \otimes l_1(x)$. Thus the set of linearized polynomials over GF$(q^m)$ with polynomial addition and the multiplication $\otimes$ forms a noncommutative ring, denoted by $L[x]$. Note that there is no left or right divisor of zero in $L[x]$\cite{ore_ams33}.

\subsection{Gabidulin Codes and Loidreau's Reconstruction Algorithm}\label{sec: GB}
The rank of a vector $\mathbf{x} \in$ GF$(q^m)^n$ is the \textbf{maximal} number of coordinates that are linearly independent over GF$(q)$, denoted as $r(\mathbf{x}; q)$. The rank distance between two vectors $\mathbf{x, y} \in$ GF$(q^m)^n$ is defined to be
\begin{equation}
d_r(\mathbf{x}, \mathbf{y}) = r(\mathbf{x-y}; q).
\end{equation}
The minimum rank distance of a code $\mathcal{C}$, denoted as  $d_r(\mathcal{C})$, is simply the minimum rank distance
over all possible pairs of distinct codewords, that is, $d_r(\mathcal{C}) = \min\limits_{\mathbf{x}_i \neq \mathbf{x}_j\in \mathcal{C}} d_r(\mathbf{x}_i, \mathbf{x}_j)$.

The maximum cardinality of a rank metric code in GF$(q^m)^n$ with minimum rank distance $d$ is $\min\{q^{m(n-d+1)}, q^{n(m-d+1)}\}$ \cite{delsarte_jct78, gabidulin_pit0185, roth_it91}. We refer to codes with maximum cardinality as maximum rank distance (MRD) codes.
A family of linear MRD codes was proposed by Gabidulin \cite{gabidulin_pit0185}, and is often referred to as Gabidulin codes.
An $(n, k)$ Gabidulin code $\mathcal{C_R}$ over GF$(q^m)$ ($n \le m$) is defined by a generator matrix $\mathbf{G}$ of the form
\begin{equation} \label{equ: Ggenematrix}
\mathbf{G} = \left( \begin{array}{cccc}
g_0   &   g_1   &  \cdots  & g_{n-1} \\
g_0^{[1]}   & g_1^{[1]}   &   \cdots   &  g_{n-1}^{[1]} \\
\vdots   &   \vdots   & \ddots  &   \vdots \\
g_0^{[k-1]}  &  g_1^{[k-1]}  &  \cdots  & g_{n-1}^{[k-1]}
\end{array} \right),
\end{equation}
where $g_0, g_1, \ldots, g_{n-1}$ are linearly independent over GF$(q)$. We introduce the vector $\mathbf{g} = (g_0, g_1, \ldots, g_{n-1})$ for future reference. For a message vector $\mathbf{u} = (u_0, u_1, \ldots, u_{k-1})$ and its corresponding message polynomial $f(x) = \sum_{i = 0}^{k-1} u_i x^{[i]}$, the codeword to be transmitted is $\mathbf{x} = (f(g_0), f(g_1), \ldots, f(g_{n-1}))$. Suppose an additive error $\mathbf{e} = (e_0, e_1, \ldots, e_{n-1})$ occurs, and the received vector is $\mathbf{y} = \mathbf{x} + \mathbf{e}= (y_0, y_1, \ldots, y_{n-1})$, where $y_i = x_i + e_i$ for $i = 0, 1, \ldots, n-1$. Given $\mathbf{y}$, a bounded distance decoder with decoding radius $t \le (n-k)/2$ tries to find $\mathbf{x'} \in \mathcal{C_R} \textrm{ and } \mathbf{e'} \in$ GF$(q^m)^n$ such that $\mathbf{y} = \mathbf{x'} + \mathbf{e'}$ with $d_r(\mathbf{y}, \mathbf{x'}) \le t$. If such $\mathbf{x'}$ and $\mathbf{e'}$ exist, the received vector $\mathbf{y}$ is said to be decodable~\cite{gabidulin_pit0185}.

Gabidulin codes can be defined using evaluation of linearized polynomials, analogous to RS codes, which are defined using evaluation of polynomials. Hence Loidreau devised a method to decode Gabidulin codes through \emph{reconstruction of linearized polynomials}\makeatletter
    \renewcommand\@cite [1]{#1}
    \makeatother (cf. [\cite{loidreau_wcc05}, Table~1]),\makeatletter
    \renewcommand\@cite [1] {[#1]}
    \makeatother
where a pair of linearized polynomials, $V(y)$ and $N(x)$ are constructed such that $V(y_i) = N(g_i)$ for $i = 0, 1, \ldots, n-1$, with deg$_q(V(y)) \le t$ and deg$_q(N(x)) \le k+t-1$. It is shown~\cite{loidreau_wcc05} that if $t \le (n - k)/2$, one gets a solution of decoding Gabidulin codes from any solution of the reconstruction problem. Loidreau's algorithm~\cite{loidreau_wcc05} constructs two sequences of polynomials $(V_0(y), N_0(x))$ and $(V_1(y), N_1(x))$, and updates them iteratively by discrepancy-based update rules, so that each sequence satisfies the objective equation for the first $i$ points after the $i$th iteration. To implement the degree constraints on the linearized polynomials, Loidreau's algorithm starts with initial polynomials of designated $q$-degrees, and then aims to increase the $q$-degrees of each sequence of polynomials strictly once every two iterations. The algorithm outputs $N_1(x)$ with $q$-degree no more than $k+\lfloor (n - k)/2 \rfloor -1$ and $V_1(y)$ of $q$-degree no more than $\lfloor (n - k)/2 \rfloor$.

\subsection{KK Codes and Their Decoding Algorithm}
\label{sec: introKK}
KK codes~\cite{kotter_it08} are a type of subspace codes for random linear network coding, where subspaces are transmitted and received at both ends. Suppose $W$ is a vector space over GF$(q)$, and $\mathcal{P}(W)$ is the set of all subspaces of $W$. For $U, V \in \mathcal{P}(W)$, the subspace distance $d_s$~\cite{kotter_it08} between $V$ and $U$ is defined as
\begin{equation} \label{equ: ds}
d_s(V, U) \df \textrm{dim}(V + U) - \textrm{dim}(V \cap U),
\end{equation}
where dim$(A)$ denotes the dimension of a subspace $A\in \mathcal{P}(W)$, $V \cap U$ is the intersection space of $V$ and $U$, and $V+U$ is the smallest subspace that contains both $V$ and $U$.

Suppose an $l$-dimensional subspace $V\in \mathcal{P}(W)$ is a codeword of a KK code. The basis of $V$ is obtained via evaluation of linearized polynomials. First we select $l$ ($l \le m$) elements $\alpha_0, \alpha_1, \ldots, \alpha_{l - 1} \in$ GF$(q^m)$ that are linearly independent over GF$(q)$. Theses $l$ elements span an $l$-dimensional vector space $\langle A \rangle \subseteq$ GF$(q^m)$, where $A = \{\alpha_i: i = 0, 1, \ldots, l -1\}$. We then construct $W$ by $W = \langle A \rangle \oplus$ GF$(q^m) = \{(\alpha, \beta): \alpha \in \langle A \rangle, \beta \in$ GF$(q^m)\}$. Given a message vector $\mathbf{u} = (u_0, u_1, \ldots, u_{k-1})$ over GF$(q^m)$, the message polynomial is defined to be $u(x) = \sum_{i = 0}^{k-1} u_i x^{[i]}$. Finally, the subspace spanned by $\{(\alpha_i, \beta_i): \beta_i = u(\alpha_i), i = 0, 1, \ldots, l-1\}$ is an $l$-dimensional subspace of $W$, as all the pairs $(\alpha_i, \beta_i)$ are linearly independent~\cite{kotter_it08}.

Suppose $V$ is transmitted over the operator channel~\cite{kotter_it08}, and an $(l-\rho + t)$-dimensional subspace $U$ of $W$ is received, with dim$(U \cap V) = l - \rho$ and $d_s(U, V) = \rho + t$. It is proved that the error is decodable by the list-1 decoding algorithm~\cite{kotter_it08} if $\rho + t < l - k + 1$.
Let $l - \rho + t = r$, and $\{(x_0, y_0), (x_1, y_1), \ldots, (x_{r-1}, y_{r-1})\}$ be a basis for $U$. The decoding algorithm in~\cite{kotter_it08} consists of an interpolation step and a factorization step. First the interpolation procedure finds a nonzero bivariate polynomial $Q(x, y)= Q_x(x) + Q_y(y)$ such that
\begin{equation} \label{equ: KKinterpolation}
Q(x_i, y_i) = 0 \textrm{ for } i = 0, 1, \ldots, r-1,
\end{equation}
where $Q_x(x)$ and $Q_y(y)$ are linearized polynomials of $q$-degrees at most $\tau - 1$ and $\tau - k$ respectively. Then a message polynomial $\hat{u}(x)$ is obtained in the factorization step by right division~\cite{kotter_it08} if $Q(x, \hat{u}(x)) \equiv 0$. Decodability is guaranteed if we select $\tau = \lceil (r + k) /2 \rceil$~\cite{kotter_it08}.

The interpolation procedure of the decoding algorithm in~\cite{kotter_it08}, called  a Sudan-style list-1 decoding algorithm,  adopts some discrepancy based update rules. During the $i$-th iteration, the algorithm generates an $x$-minimal bivariate polynomial and a $y$-minimal bivariate polynomial, $f_0^{(i)}(x,y)$ and $f_1^{(i)}(x,y)$, that interpolate through the first $i$ points for $i=1,2,\ldots, r$, where $r$ is the total number of points to be interpolated. Finally, the minimum one between $f_0^{(r)}(x,y)$ and $f_1^{(r)}(x,y)$, defined under an order of $\prec$~\cite{kotter_it08}, is the decoding output.

%

\subsection{MV Codes and Their List Decoding Algorithm}
\label{sec: introMV}
MV codes are similar to but different from KK codes~\cite{kotter_it08}. To enable list decoding, different code constructions are proposed for different code dimensions in~\cite{mahdavifar_isit10}~\cite{mahdavifar_it10}.

To construct an $l$-dimensional MV code over GF$(q^{ml})$, $l$ has to be a positive integer that divides $q-1$. Then the equation $x^l - 1 = 0$ has $l$ distinct roots $e_1  = 1, e_2, \ldots, e_l$ over GF$(q)$. Choose a primitive element $\gamma$ over GF$(q^{ml})$ with $\gamma, \gamma^{[1]},\ldots,\gamma^{[ml-1]}$ being a normal basis for GF$(q^{ml})$. Then construct elements $\alpha_i$ over GF$(q^{ml})$ by $\alpha_i = \gamma + e_i \gamma^{[m]} + e_i^2 \gamma^{[2m]} + \cdots + e_i^{l-1}\gamma^{[m(l-1)]}$ for $i = 1,2, \ldots, l$. It is proved~\cite{mahdavifar_it10} that the set $\{ \alpha_i^{[j]} : i = 1,2,\ldots, l, j = 0,1, \ldots, m-1\}$ is a basis of GF$(q^{ml})$ over GF$(q)$.

For a message vector $\mathbf{u} = (u_0, u_1, \ldots, u_{k-1})$ over GF$(q)$, the message polynomial is $u(x) = \sum_{i=0}^{k-1} u_i x^{[i]}$. Let $u^{\otimes i}(x)$ denote the composition of $u(x)$ with itself by $i$ times for any nonnegative integer $i$, while $u^{\otimes 0 }(x) =x$. Then the codeword $V$ corresponding to the message $\mathbf{u}$ is spanned by a set of vectors $v_i$ for $i = 1,2,\ldots, l$, where $v_1 = (\alpha_1, u(\alpha_1), u^{\otimes 2}(\alpha_1), \ldots, u^{\otimes L}(\alpha_1))$, $v_i = (\alpha_i, \frac{u(\alpha_i)}{\alpha_i},\ldots, \frac{u^{\otimes L}(\alpha_i)}{\alpha_i})$, and $L$ is the desired list size. Note that $\frac{u^{\otimes j}(\alpha_i)}{\alpha_i} \in$ GF$(q^m)$ for any $j \ge 0$ and $i=2,3,\ldots, l$~\cite{mahdavifar_it10}. Then $V$ is an $l$-dimensional subspace of the $(Lm+l)$-dimensional ambient space $W = \langle \alpha_1, \alpha_2,\ldots, \alpha_l \rangle \oplus \underbrace{ \textrm{GF}(q^m) \oplus \cdots \oplus \textrm{GF}(q^m)}_{L\textrm{ times}}$. Suppose an error of dimension $t$ occurs, and an $(l+t)$-dimensional subspace $U$ of $W$ is received. The decoder first finds subspaces $U_i$ such that $U_i = \{ (x, y_1, y_2, \ldots, y_L) : x \in \langle \alpha_i \rangle\}$ for $i=1,2,\ldots,l$. Then, a basis $\{(x_{1,j}, y_{1,1,j}, y_{1,2,j},\ldots, y_{1, L, j}): j = 1, 2, \ldots, r_1\}$ of $U_1$ is found, where $r_1$ is the dimension of $U_1$. If $l=1$, we ignore the first step and simply find a basis for the $(t+1)$-dimensional received subspace $U_1$. For $i = 2, 3, \ldots, l$, the decoder obtains $U'_i = \{ (x, \alpha_i y_1, \alpha_i y_2, \ldots, \alpha_i y_L) : (x, y_1, y_2, \ldots, y_L) \in U_i\}$, and finds a basis $\{(x_{i,j}, y_{i,1,j}, y_{i,2,j},\ldots, y_{i, L, j}): j = 1, 2, \ldots, r_i\}$ of $U'_i$, where $r_i$ is the dimension of $U_i$. Finally, the decoder constructs a nonzero multivariate polynomial $Q(x, y_1, y_2, \ldots, y_L)= Q_0(x) + Q_1(y_1) + Q_2(y_2) + \cdots + Q_L(y_L)$, where $Q_s$ is a linearized polynomials over GF$(q^{ml})$ of $q$-degree at most $ml -s(k-1) - 1$ for $s = 0, 1, \ldots, L $, such that for $i= 1,2,\ldots, l$, $j = 1,2, \ldots, r_i$, and $h = 0,1,\ldots, m-1$,
\begin{equation} \label{equ: lLinter}
Q(x_{i,j}^{[h]}, y_{i,1,j}^{[h]}, \ldots, y_{i, L, j}^{[h]}) = 0.
\end{equation}
Using the LRR algorithm in~\cite{mahdavifar_isit10}, the decoder finds all possible polynomials $\hat{u}(x)$'s such that
\begin{equation} \label{equ: lLfactor}
Q(x, \hat{u}(x), \hat{u}^{\otimes 2}(x),\ldots,\hat{u}^{\otimes L}(x)) \equiv 0.
\end{equation}
It is proved~\cite{mahdavifar_it10} that~(\ref{equ: lLinter}) has a nonzero solution if
\begin{equation} \label{equ: lLcondition}
t < lL -L(L+1)\frac{k-1}{2m},
\end{equation}
and there are at most $L$ solutions to (\ref{equ: lLfactor}), among which the transmitted message polynomial $u(x)$ is guaranteed to be included.

\section{General Interpolation by Linearized Polynomials}
\label{sec:generalinterpo}

In this section, we investigate the general interpolation problem by linearized polynomials. We first present the general interpolation problem, then propose our general interpolation algorithm, which follows a strategy similar to that in~\cite{WangMcEliece05}.

\subsection{General Interpolation over Free $L[x]$-Modules} \label{section: GeneralInterpolation}

Suppose $L[x]$ is the ring of linearized polynomials over GF$(q^m)$, and $V$ is a free $L[x]$-module with a basis $B=\{b_{0}, b_{1}, \ldots, b_{L} \}$.  We denote the multiplication between an element in $L[x]$ and an element in the module by $\circ$, and any element $Q\in V$ can be represented by
\begin{equation} \label{equ: GIexpand}
Q = \sum_{j=0}^L l_{j}(x) \circ b_{j} = \sum_{j=0}^{L} \sum_{i \geq 0} a_{i,j}x^{[i]} \circ b_{j},
\end{equation}
where $l_{j}(x) \in L[x]$ and $a_{i,j} \in \textrm{GF}(q^m)$. Thus $V$ is also a vector space over GF$(q^m)$ with a basis
\begin{equation} \label{equ: GIM}
M = \{ x^{[i]} \circ b_{j}, i\ge 0, j = 0, 1, \ldots, L\}.
\end{equation}
Suppose there exists a total ordering $<$ on $M$, and we can write $M = \{\phi_j\}_{j \ge 0}$ such that $\phi_i < \phi_j$ when $i < j$. Then $Q \in V$ can be represented by
\begin{equation} \label{equ: orderM}
Q = \sum_{j = 0}^{J} a_{j} \phi_{j},
\end{equation}
where $\phi_{j} \in M$ and $a_{J} \ne 0$. $J$ is called the \emph{order} of $Q$, denoted as order$(Q)$, and $\phi_J$ is the \emph{leading monomial} of $Q$, denoted as LM$(Q)$. We write $Q <_o Q'$ if order$(Q) <$ order$(Q')$, and $Q =_o Q'$ if order$(Q) =$ order$(Q')$. An element $Q$ is a minimum in a subset of $V$ if its order is the lowest among all the elements in the subset. Further, we define Ind$_y(l(x)\circ b_j) = j$, and Ind$_y(Q) = $ Ind$_y($LM$(Q))$. Then we introduce a partition of $V$ as $V = \bigcup_jS_j$, where $S_j = \{Q \in V: \textrm{Ind}_y(Q) = j \}$.

For the vector space $V$ over GF$(q^m)$, we consider a set of $C$ linear functionals $D_{i}$ from $V$ to GF$(q^m)$: $D_1, D_2, \ldots, D_{C}$. Suppose $K_i$ is the kernel of $D_i$ and $\overline{K}_{i} = K_{1} \cap K_{2} \cap \cdots \cap K_{i}$ is an $L[x]$-submodule, then the general interpolation problem is to find a minimum $Q^* \in \overline{K}_{C}$, that is, to find an element $Q^* \in V$ such that it lies in the kernels of all the given linear functionals. Furthermore, we can show the uniqueness of $Q^*$ as in~\cite{WangMcEliece05}.
\begin{lemma} \label{lemma: uniqueness}
The minimum in $\overline{K}_C$ is unique up to a scalar.
\end{lemma}
\begin{proof}
Suppose both $Q^*$ and $Q'$ have the minimum order in $\overline{K}_C$, then there exists a nontrivial linear combination $\alpha Q^* + \beta Q'$, where $\alpha, \beta \in $ GF$(q^m)$, such that $\alpha Q^* + \beta Q' <_o Q^* =_o Q'$. Since $\alpha Q^* + \beta Q' \in \overline{K}_C$, this contradicts the minimality of $Q^*$ and $Q'$.
\end{proof}

\begin{algorithm}
\caption{General Interpolation by Linearized Polynomials}
\label{alg: GILinearized}
\begin{algorithmic}
\FOR{$j = 0$ to $L$}
\STATE {$g_{0,j} \leftarrow  b_j$}
\ENDFOR
\FOR{$i = 0$ to $C-1$}
\FOR{$j = 0$ to $L$}
\STATE{$g_{i+1, j} \leftarrow g_{i,j}$}
\STATE{$\Delta_{i+1,j}  \leftarrow D_{i+1}(g_{i,j})$}
\ENDFOR
\STATE{$J \leftarrow \{ j: \Delta_{i+1,j} \ne 0\}$}
\IF{$J \ne \emptyset$}
\STATE{$ j^* \leftarrow\underset{j \in J}{\mathrm{argmin}}\{g_{i,j} \}$}
\FOR{ $j \in J$ }
\IF{$j \ne j^*$}
\STATE{$g_{i+1,j} \leftarrow \Delta_{i+1,j^*} g_{i,j} - \Delta_{i+1,j} g_{i,j^*}$}
\ELSIF{$j = j^*$}
\STATE{$g_{i+1,j} \leftarrow \Delta_{i+1,j} (x^{[1]} \circ g_{i,j}) - D_{i+1}(x^{[1]} \circ g_{i,j})g_{i,j}$}
\ENDIF
\ENDFOR
\ENDIF
\ENDFOR
\STATE{$Q^* \leftarrow \min\limits_{j} ~g_{C,j}$}
\end{algorithmic}
\end{algorithm}

Define $T_{i,j} = \overline{K}_i \cap S_j$, and $g_{i,j} = \min \limits_{g \in T_{i,j}} ~g$, then the general interpolation problem is equivalent to finding $g_{i,j}$ for $i = 1, \ldots, C$. The key idea is to iteratively construct $g_{i+1,j}$ from $g_{i,j}$ by a discrepancy based update, starting from some initial values. We propose an algorithm to solve this general interpolation algorithm, given in  Algorithm~\ref{alg: GILinearized}. In Algorithm~\ref{alg: GILinearized}, there are three cases, and in each case a different update is used to obtain $g_{i+1, j}$ based on $g_{i, j}$.
\begin{enumerate}
\item If $g_{i,j} \in K_{i+1}$, then $g_{i+1,j} = g_{i,j}$.
\item For $g_{i,j}$'s not in $K_{i+1}$, we find one of them with the lowest order, denoted as $g_{i,j^*}$, and check whether $g_{i,j} =_o g_{i,j^*}$. If $g_{i, j} \ne_o g_{i, j^*}$, then $g_{i+1,j} = D_{i+1}(g_{i,j^*}) g_{i,j} - D_{i+1}(g_{i,j}) g_{i,j^*}$. We call this type of update a cross-term update. Note in this case, the order of $g_{i,j}$ is preserved, that is, $g_{i+1, j} =_o g_{i,j}$.
\item For $g_{i+1, j^*}$, we construct $g_{i+1,j^*}$ by $g_{i+1,j^*} = D_{i+1}(g_{i,j^*}) (x^{[1]} \circ g_{i,j^*}) - D_{i+1}(x^{[1]} \circ g_{i,j^*})g_{i,j^*}$. We call this type of update an order-increase update. In this case, $g_{i+1, j^*}$ takes a higher order than $g_{i,j^*}$, that is, $g_{i, j^*} <_o g_{i+1,j^*}$.
\end{enumerate}

\begin{lemma} \label{lemma: GIminimum}
In each of the three cases, $g_{i+1, j}$ is a minimum in $T_{i+1, j}$.
\end{lemma}
The following proof follows the same approach in the corresponding proof of minimality in~\cite{McEliece03} and~\cite{WangMcEliece05}.
\begin{proof}
We deal with the three cases separately:
\begin{enumerate}
\item We choose $g_{i+1, j} = g_{i,j}$ if $g_{i,j} \in T_{i+1, j}$. Since $g_{i,j}$ is a minimum in $T_{i,j}$ and $T_{i, j} \supseteq T_{i+1,j}$, $g_{i,j}$ is also a minimum in the smaller set $T_{i+1,j}$.
\item If $g_{i, j^*} <_o g_{i,j}$, we construct $g_{i+1, j} = D_{i+1}(g_{i,j^*}) g_{i,j} - D_{i+1}(g_{i,j}) g_{i,j^*}$. One can verify that $D_{i+1}(g_{i+1, j}) = 0$, and thus $g_{i+1,j} \in K_{i+1}$. Since Ind$_y (g_{i+1, j})$ = Ind$_y (g_{i, j})$, $g_{i+1, j}$ is also in $S_j$, thus $g_{i+1, j} \in T_{i+1, j}$.  Furthermore, $D_k(g_{i+1, j}) = D_{i+1}(g_{i,j^*}) D_k(g_{i,j}) - D_{i+1}(g_{i,j}) D_k(g_{i,j^*}) = 0$ for any $k \le i$, since $g_{i,j}, g_{i, j^*} \in \overline{K_i}$. Hence $g_{i+1, j} \in T_{i, j}$. Since $g_{i+1, j} =_o g_{i, j}$ and $g_{i, j}$ is a minimum in $T_{i, j}$, $g_{i+1, j}$ is also a minimum in $T_{i, j}$, hence a minimum in the smaller set $T_{i+1, j}$.

\item In this case, $g_{i+1,j^*} = D_{i+1}(g_{i,j^*}) (x^{[1]} \circ g_{i,j^*}) - D_{i+1}(x^{[1]} \circ g_{i,j^*})g_{i,j^*}$. First note that $D_{i+1}(g_{i+1, j^*}) = 0$, and hence $g_{i+1, j^*} \in K_{i+1}$. For any $k \le i$, when we apply $D_k$ to $g_{i+1, j^*}$, we also get zero because both $g_{i,j}$ and $x^{[1]} \circ g_{i,j^*}$ lie in $\overline{K}_i$, as $K_i$ is a submodule of $L[x]$. Thus $g_{i+1, j^*} \in \overline{K}_{i+1}$. Also, Ind$_y (g_{i+1, j^*})=$ Ind$_y (x^{[1]} \circ g_{i,j^*}) = j^*$ by our definition Ind$_y(l(x)\circ b_j) = j$. Thus we have $g_{i+1, j^*} \in T_{i+1, j^*}$. Next we show that $g_{i+1,j^*}$ is a minimum in $T_{i+1, j^*}$ by contradiction. Suppose there exists $f_{i+1, j^*} \in T_{i+1, j^*}$ such that $f_{i+1, j^*} <_o g_{i+1, j^*}$. Note that order$(g_{i+1, j^*})=$ order$(x^{[1]}\circ g_{i,j^*})$. Since $T_{i, j^*} \supseteq T_{i+1,j^*}$, $f_{i+1, j^*}$ also lies in $T_{i, j^*}$. Hence order$(f_{i+1, j^*}) \ge$ order$(g_{i, j^*})$, as $g_{i, j^*}$ is a minimum in $T_{i, j}$, which results in order$(g_{i, j^*}) \le$ order$(f_{i+1, j^*}) <$ order$(x^{[1]}\circ g_{i,j^*})$. Since both $g_{i, j^*}$ and $x^{[1]}\circ g_{i,j^*}$ lie in the set $S_{j^*}$ by definition, there does not exist $f_{i+1, j^*} \in S_{j^*}$ such that order$(g_{i, j^*}) <$ order$(f_{i+1, j^*}) <$ order$(x^{[1]}\circ g_{i,j^*})$. Hence the only possibility is that $f_{i+1, j^*} =_o g_{i, j^*}$. But in this case, we could construct $h = \alpha f_{i+1, j^*} + \beta g_{i, j^*}$ with $\alpha, \beta \in $GF$(q^m)$ such that $h <_o g_{i, j^*}$. Note that $h \in \overline{K}_i$ but $h \notin \overline{K}_{i+1}$ as $f_{i+1, j^*} \in T_{i+1}$ but $g_{i, j^*} \notin T_{i+1}$. The fact that $h \in \overline{K}_i \backslash \overline{K}_{i+1}$ but $h <_o g_{i, j^*}$ contradicts the minimality of $g_{i, j^*}$ in $\overline{K}_i \backslash \overline{K}_{i+1}$, as $g_{i, j^*}$ has the lowest order among all $g_{i,j}$'s  where $g_{i, j} \in \overline{K}_i$ but $g_{i,j} \notin \overline{K}_{i+1}$.
\end{enumerate}
\end{proof}

\subsection{Complexity Analysis of Algorithm~\ref{alg: GILinearized}} \label{section: ComplexityGeneral}
There are a total of $C$ iterations in Algorithm~\ref{alg: GILinearized}. In each iteration, $L+1$ linear functionals are first carried out to calculate the discrepancies, followed by at most $L$ finite field additions (subtractions) to find the minimum candidate and its index among those with nonzero discrepancies. Then to update the candidates, we conduct at most $2(L+1)^2 (D+1)$ finite field multiplications, $(L+1)^2 (D+1)$ finite field additions, one multiplication between elements in the ring $L[x]$ and elements in the module $V$, and one computation of the linear functional, where $D$ is the highest $q$-degree of the linearized polynomials in $x$ among all the iterations. Notice that the $q$-degree of each candidate is non-decreasing in an iteration based on the update rules. Hence it is safe to choose $D$ to be the highest $q$-degree of the polynomial in $x$ of the ultimate output. To sum up, the complexity of Algorithm~\ref{alg: GILinearized} is dominated by $O(CDL^2)$ finite field additions, $O(CDL^2)$ field multiplications, $O(CL)$ linear functional calculations, and $O(C)$ multiplications between elements in the ring $L[x]$ and elements in the module $V$. Since the complexity of the  linear functional calculations and the multiplications between elements in the ring and elements in the module might vary in different situations, we consider the complexity of each realization of Algorithm~\ref{alg: GILinearized} on a case-by-case basis.

\section{Decoding of Gabidulin Codes}
\label{sec: GeneralGb}
\subsection{Decoding of Gabidulin Codes}
We consider an $(n,k)$ Gabidulin code over GF$(q^m)~(n \le m)$ as defined in Section~\ref{sec: GB}, and the ring of linearized polynomials $L[x]$ over GF$(q^m)$ discussed in Section~\ref{sec: L[x]}. Based on Loidreau's polynomial reconstruction approach~\cite{loidreau_wcc05}, we generalize the decoding problem of Gabidulin codes from an interpolation point of view. Suppose we have a set of points $(x_i, y_i)$ with $y_i = f(x_i) + e_i$ for $i = 0, 1, \ldots, n-1$, where $x_i$'s are linearly independent and $r(\mathbf{e}; q) \le t$. Try to construct a nonzero bivariate polynomial $Q(x, y) = Q_0(x) + Q_1(y)$ with $Q_1(x)$ and $Q_2(y)$ being linearized polynomials over GF$(q^m)$, such that $\max \{\textrm{deg}_q (Q_1(x)), k - 1 + \textrm{deg}_q(Q_1(y))\}$ is as small as possible and
\begin{equation} \label{equ: GGeneralInter}
Q(x_i, y_i)  = Q_1(x_i) +Q_2(y_i) = 0 \textrm{ for } i = 0, 1, \ldots, n-1.
\end{equation}
We will show that a solution of~(\ref{equ: GGeneralInter}) gives a solution to the decoding problem of Gabidulin codes under some conditions. Then we formalize~(\ref{equ: GGeneralInter}) to a general interpolation problem over free $L[x]$-modules, and solve it by Algorithm~\ref{alg: GILinearized}.

Suppose deg$_q(V) = \tau$, and deg$_q(N) = \tau + k - 1$. To have a nonzero solution of~(\ref{equ: GGeneralInter}), the number of unknown coefficients should be greater than the number of equations, that is,
\begin{equation} \label{equ: tau1mod}
2\tau > n - k - 1.
\end{equation}
Next suppose $Q(x, y) = N(x) - V(y)$ is a nonzero solution of~(\ref{equ: GGeneralInter}). Substituting $y$ by the message polynomial $f(x)$, we get $Q(x, f(x)) = N(x) - V(f(x))$. When $Q(x, f(x)) \equiv 0$, i.e., $N(x) - V(f(x))$ is the zero polynomial, $f(x)$ satisfies $N(x) = V(x) \otimes f(x)$ and thus can be obtained by right division over the linearized polynomial ring~\cite{kotter_it08}.

It remains to identify the condition under which $Q(x, f(x))$ is identically zero. Since $Q(x, y) = N(x) - V(y)$ is a nonzero solution of (\ref{equ: GGeneralInter}), $Q(x_i, y_i) = N(x_i) - V(y_i) = 0$, i.e., $N(x_i) - V(f(x_i)) = V(e_i)$ with the rank of $(V(e_0), V(e_1), \ldots, V(e_{n-1}))$ no more than $t$. Then there exists a nonzero linearized polynomial $W$ of $q$-degree at most $t$ such that $W(V(e_i)) = W(N(x_i) - V(f(x_i))) = 0$ for $i = 0, 1, \ldots, n-1$. Then we have a linearized polynomial $W(N(x) - V(f(x)))$ of $q$-degree at most $t + \tau + k -1$ with $n$ linearly independent roots $x_i$ for $i = 0, 1, \ldots, n-1$. Thus when $t + \tau + k -1 < n$, we have $W(N(x) - V(f(x))) \equiv 0$. Since there is no left or right divisor or zero in the linearized polynomial ring~\cite{ore_ams33} and $W$ is nonzero, we have $N(x) - V(f(x)) \equiv 0$, hence $f(x)$ can be obtained by right division over the linearized polynomial ring. The condition $t + \tau + k -1 < n$ will be satisfied by forcing
\begin{equation} \label{equ: req}
t \le \tau,
\end{equation}
and restricting
\begin{equation} \label{equ: tau2mod}
2 \tau < n - k + 1.
\end{equation}
Combining~(\ref{equ: tau1mod}) and~(\ref{equ: tau2mod}), we select $\tau = \lfloor (n - k) / 2 \rfloor$, and have
\begin{equation} \label{equ: tandtau}
t \le \lfloor (n-k)/2 \rfloor.
\end{equation}
Hence if~(\ref{equ: tandtau}) is satisfied, a solution of~(\ref{equ: GGeneralInter}) gives a solution to the decoding problem of Gabidulin codes. Next we formalize the interpolation problem in (\ref{equ: GGeneralInter}) to a general interpolation problem over free $L[x]$-modules.

We select $B = \{ b_0, b_1\} = \{ x, y\}$ as a basis, and construct a free $L[x]$-module $V = \{Q(x,y)\}$ from
\begin{equation} \label{eq: Y_basisGB}
Q(x,y) = l_0(x) \circ b_0 + l_1(x) \circ b_1,
\end{equation}
where $l_0(x), l_1(x) \in L[x]$, and the multiplication $\circ$ is defined as
\begin{equation} \label{equ: GIcirc}
l(x) \circ b_j \df l(b_j),  \textrm{  for } j = 0, 1.
\end{equation}
Hence $Q(x,y)=l_0(x) + l_1(y)$, and we call such $Q(x,y) \in V$ a  \emph{bivariate linearized polynomial}. Following~(\ref{equ: GIexpand}) and~(\ref{equ: GIM}), $V$ is also a vector space over GF$(q^m)$ with a vector space basis $M =  \{ x^{[i]} \circ b_{j}, i\ge 0, j = 0, 1\}$. Then we define a total ordering on $M$ as follows. We write $b_j^{[i]} < b_j^{[i+1]}$ for $j \in \{0,1\}$ and $i \ge 0$, and write $x^{[i+k-1]} < y^{[i]} < x^{[i+k]}$ for $i \ge 0$. Once the total ordering on $M$ is determined, the leading monomial and the order of any $Q\in V$ can be defined as described in Section~\ref{section: GeneralInterpolation}. Consequently, given a subset of $V$, a minimum element in $V$ can be found.

Finally, we define a set of linear functionals $D_i$ from $V$ to GF$(q^m)$ to be $D_i(Q) = Q(x_{i}, y_{i}) = l_0(x_i) + l_1(y_i)$ for $i = 0, 1, \ldots, n-1$, where $(x_i, y_i)$'s are the points to be interpolated. If $D_i(Q(x, y)) = 0$, $Q(x, y)$ is said to be in the kernel $K_i$ of $D_i$. The kernels are $L[x]$-submodules by the following lemma.

\begin{lemma} \label{lemma: LxKernels}
$K_i$ is an $L[x]$-submodule.
\end{lemma}
\begin{proof}
Since $K_i$ is a subgroup of $V$, it is easy to show that $K_i$ is an Abelian group under polynomial addition, and the associative  and distributive laws hold for the multiplication between elements in $L[x]$ and elements in $K_i$. Now we consider an arbitrary element $l(x) \in L[x]$ and an element $Q(x, y) \in K_i$. Suppose $l(x) = \sum_{j=0}^{n} a_j x^{[j]}$ with $a_i \in $ GF$(q^m)$, and $Q(x, y) = l_0(x) + l_1(y)$ with $l_0(x), l_1(x) \in L[x]$. Then $l(x) \circ Q(x, y) = (\sum_{j=0}^{n} a_j x^{[j]}) \circ (l_0(x) + l_1(y)) = \sum_{j=0}^{n} a_j (l_0(x)^{[j]} + l_1(y)^{[j]}) = \sum_{j=0}^{n} a_j (l_0(x) + l_1(y))^{[j]}$. Given $D_i(Q(x, y)) = l_0(x_i) + l_1(y_j) =0$, we have  $D_i(l(x) \circ Q(x, y)) = \sum_{j=0}^{n} a_j (l_0(x_i) + l_1(y_i))^{[j]} = 0$. Hence $K_i$ is a submodule of $L[x]$.
\end{proof}

Hence $\overline{K}_i$ is also an $L[x]$-submodule. Consequently, the interpolation problem described by (\ref{equ: GGeneralInter}) is to find a minimum $Q\in V$ such that $Q$ is a minimum in $\overline{K}_{n-1}$. This is a general interpolation problem over free $L[x]$-modules as described in Section~\ref{section: GeneralInterpolation}, and Algorithm~\ref{alg: GILinearized} solves it by finding a \textbf{minimum} nonzero solution.

To use Algorithm~\ref{alg: GILinearized}, first we set $g_{0,0} = x$, and $g_{0,1} = y$ in the initialization step. In the following iterations, multiplication between an element in $L[x]$ and an element in $V$ in the cross-term and order-increase updates follow (\ref{equ: GIcirc}). In particular, $g_{i+1,j^*} = D_{i+1}(g_{i, j^*}) (x^{[1]}\circ g_{i, j^*}) - D_{i+1}(x^{[1]}\circ g_{i, j^*} ) g_{i, j^*}$. Since $D_{i+1}(g_{i, j^*}) \ne 0$, we can also omit it from the right hand side, and instead use $g_{i+1,j^*} = g_{i, j^*}^q - (D_{i+1}(g_{i, j^*}))^{q-1} g_{i, j^*}$, as scaling by a nonzero scalar does not affect the order of an element in $V$.

Now we consider the complexity of Algorithm~\ref{alg: GILinearized} when used to decode Gabidulin codes. Adopting the same set of parameters in the complexity analysis in Section~\ref{section: ComplexityGeneral}, we have $L = 1$, $C = n$, and $D = \lfloor \frac{n+k}{2} \rfloor$ based on~(\ref{equ: KKinterpolation}) and the following argument. Second, each linear functional in this case carries out evaluations of the bivariate linearized polynomial by the given points, with a total of $O(DL)$ finite field multiplications and  $O(DL)$ finite field additions.
Finally, the multiplication between $x^{[1]}$ and $g_{i,j}$ is accomplished by raising the coefficients of $g_{i,j}$ to the $q$-th power, which is simply a cyclic shift if a normal basis is chosen~\cite{silva_isit09}\cite{max_ciss08}.
In summary, when used to decode KK codes, Algorithm~\ref{alg: GILinearized} takes a total of $O(n(n+k))$ finite field multiplications in GF$(q^m)$.  On the other hand, the complexity analysis in \cite{Loidreau_wcc06} gives an overall complexity of $\frac{5}{2}n^2 - \frac{3}{2}k^2 + \frac{n-k}{2}$. Hence both algorithms are of quadratic complexity.

\subsection{Comparison to Loidreau's Reconstruction Algorithm}

Although our cross-term and order-increase update rules are similar to that of the alternate increasing degree step in Loidreau's algorithm, we observe that Algorithm~\ref{alg: GILinearized} differs from Loidreau's algorithm in two aspects, stated as follow.

First, Loidreau's algorithm uses another algorithm~\cite{ore_ams34} in the precomputation step before initializing the main algorithm, for the purpose of reduced complexity, whereas our decoding algorithm carries out all the iterations solely from the interpolation approach. But as shown in the previous section, both Loidreau's algorithm and Algorithm \ref{alg: GILinearized} have quadratic complexities. Further, we will show the equivalence of the polynomials derived after the initialization step of Loidreau's algorithm and the ones obtained after the first $k$ iterations of Algorithm~\ref{alg: GILinearized}. The initialization step of Loidreau's algorithm actually introduces two bivariate polynomials $Q_0 = N_0(x) - V_0(y)$ and $Q_1 = N_1(x) - V_1(y)$. Given our previous notations, Algorithm~\ref{alg: GILinearized} produces two bivariate polynomials $g_{k, 0}$ and $g_{k,1}$ after the first $k$ iterations. The relation between these four polynomials are stated in Lemma~\ref{Lemma: LdGIkit}.

\begin{lemma} \label{Lemma: LdGIkit}
The initial bivariate polynomials of Loidreau's algorithm and the bivariate polynomials derived after the first $k$ iterations of Algorithm~\ref{alg: GILinearized} are of the same order correspondingly, i.e., $Q_0 =_o g_{k, 0}$ and $Q_1 =_o g_{k, 1}$.
\end{lemma}
\begin{proof}
 In the initialization step of Algorithm~\ref{alg: GILinearized}, $g_{0,0} =x$ is of lower order than $g_{0,1} =y$, and $D_1(g_{0,0}) = x_0 \ne 0$ as $x_i$'s are linearly independent, so $g_{0,0} = x$ updates by the order-increase rule, while $g_{0,1} =y$ updates according to its discrepancy value. Then $g_{1,0}$ is actually a linearized polynomial in $x$ of $q$-degree 1, and $g_{1,1}$ is a bivariate polynomial with a leading monomial $cy$, where $c \in$ GF$(q^m)$ is a constant.

 In the second iteration, again $g_{1,0} <_o g_{1,1}$ based on our total ordering on $M$, and $D_2(g_{1,0}) \ne 0$ as $x_i$'s are linearly independent, i.e., there does not exist a linearized polynomial of $q$-degree 1 that has two linearly independent roots. Hence $g_{1,0}$ takes the order-increase rule and $g_{1,1}$ adopts others accordingly. Similar situation occurs in all the first $k$ iterations, given the total ordering we defined on $M$ and the fact that $D_{i+1}(g_{i, 0}) \ne 0$ for any $i \le k$.

 Finally, a $g_{k,0}$ in $x$ of $q$-degree $k$ is derived, which actually only interpolates over the first $k$ $x_i$'s. Note that $N_0(x)$ is obtained in the same way by $N_0(x) = Int(x_0,\ldots, x_{k-1})$ in Loidreau's algorithm. Given that $V_0(y) = 0$,  we have $Q_0= N_0(x)$. Hence $g_{k,0} =_o Q_0$. On the other hand, $g_{k,1}$ is a bivariate polynomial with a leading monomial $c'y$, where $c' \in$ GF$(q^m)$ is also a constant. Since $N_1(x)$ is a linear combination of linearized polynomials of $q$-degree $k-1$, it is a linearized polynomial in $x$ of $q$-degree at most $k-1$, then the leading monomial of $Q_1$ is $y$. As a result, $g_{k,1} =_o Q_1$.
\end{proof}
Note that the $q$-degree of $N_0(x)$ is exactly $k$, as it actually interpolates over $k$ linearly independent points $x_0, x_1, \ldots, x_{k-1}$. $N_1(x)$ is a linear combination of polynomials of $q$-degree $k-1$, but its $q$-degree might be lower than $k-1$, as the most significant coefficients may cancel each other. Thus the claim in~\cite{loidreau_wcc05} that after the final iteration deg$_q(V_1(y)) = \lfloor (n-k) / 2\rfloor$ is inaccurate.

The second difference between Loidreau's and our decoding algorithms lies in the update of the interpolation steps when some of the discrepancies are zero. It should be pointed out that in the alternate increasing degree step of Loidreau's algorithm, $s_0$ in operations $(c)$ and $(d)$ should be $s_0^{(q-1)}$
in\makeatletter
    \renewcommand\@cite [1]{#1}
    \makeatother ([\cite{loidreau_wcc05}, Table~1]).\makeatletter
    \renewcommand\@cite [1] {[#1]}
    \makeatother
After the correction of this typo, the key difference between Loidreau's algorithm and Algorithm~\ref{alg: GILinearized} is that the latter accounts for zero discrepancies, while the former only covers it partially. To be specific, Loidreau's algorithm\makeatletter
    \renewcommand\@cite [1]{#1}
    \makeatother [\cite{loidreau_wcc05}, Table~1]\makeatletter
    \renewcommand\@cite [1] {[#1]}
    \makeatother
malfunctions when $s_1 = 0$ but $s_0 \ne 0$, as shown in Lemma~\ref{lemma: Ldzero}.
\begin{lemma} \label{lemma: Ldzero}
If $s_1 = 0$ but $s_0 \ne 0$ at the beginning of any iteration, all four linearized polynomials of the $V_0, N_0, V_1$ and $N_1$ in Loidreau's algorithm will be the zero polynomial after a certain number of iterations.
\end{lemma}
The proof can be conducted simply by tedious calculations, hence we will not present it here. Instead, an example is given to illustrate Lemma~\ref{lemma: Ldzero}, where $s_1 = 0$ but $s_0 \ne 0$ happens during an intermediate iteration. To fix the problem in Lemma~\ref{lemma: Ldzero}, one way is not to update the candidates when the zero discrepancy is involved. But such an operation breaks the rule of updating the $q$-degrees of the candidates alternately, which is designed to ensure strict degree constraints on the output of the algorithm. Further notice that $s_0$ and $s_1$ are involved in different types of update rules for the two pairs of candidate polynomials, hence for the case of $s_1 \ne 0$ but $s_0 = 0$, the algorithm in\makeatletter
    \renewcommand\@cite [1]{#1}
    \makeatother [\cite{loidreau_wcc05}, Table~1]\makeatletter
    \renewcommand\@cite [1] {[#1]}
    \makeatother works properly.

\begin{example} \label{example: 2szeros}
We construct a $(6, 2)$ Gabidulin code over GF$(2^6)$ with $\mathbf{g} = ( \alpha^{31}, \alpha^{48}, \alpha^{32},\alpha^{16},1,\alpha^{47})$, where $\alpha$ is a primitive element of GF$(2^6)$ and is a root of $x^6 + x + 1 = 0$. Given the message vector $\mathbf{u} = (1, 0)$, the message polynomial is $f(x) = x$, with a codeword $\mathbf{x} = (f(g_0), f(g_1), \ldots, f(g_{n-1})) = \mathbf{g}$. Suppose the error vector is $\mathbf{e} = ( 0, \alpha^{48}, \alpha^{54}, 0, 0, 0 )$, and the received vector is $\mathbf{y} = \mathbf{x} + \mathbf{e} = \left(\alpha^{31},0, \alpha^{19}, \alpha^{16}, 1, \alpha^{47}\right)$. The decoding procedures by Loidreau's algorithm and Algorithm~\ref{alg: GILinearized} are presented in Table~\ref{tab: LdvsGI}. Based on Lemma~\ref{Lemma: LdGIkit}, we start from the initial polynomials of Loidreau's algorithm and the polynomials after the first $k$ iterations by Algorithm~\ref{alg: GILinearized}. Note that $g_0$ in the final iteration is not listed, as it is of higher order than $g_1$. Since $r(\mathbf{e}; q) = 2 \le t = (n - k)/2$, $\mathbf{y}$ is decodable. As shown in Table~\ref{tab: LdvsGI}, however, Loidreau's algorithm fails. On the other hand, our algorithm produces a bivariate polynomial $g_{n,1} = \alpha^4 x^4 + x^2 + \alpha^{29}x + \alpha^4 y^4 + y^2 + \alpha^{29}y$, from which the correct decoding result $f(x) = x$ is obtained.
\end{example}

\begin{table*}[htbp]
\begin{center}
\caption{Example~\ref{example: 2szeros}: Use Loidreau's algorithm and Algorithm~\ref{alg: GILinearized} to decode Gabidulin codes}
\label{tab: LdvsGI}
\begin{tabular}{|l|l|l|}
\hline
$i$  &  Loidreau's algorithm   &  Algorithm~\ref{alg: GILinearized} \\ \hline
\multirow{2}{*}{2} & $N_0 = x^4 + \alpha^5 x^2 + \alpha^{31} x, V_0 = 0$ & $g_0 = x^4 + \alpha^5 x^2 + \alpha^{31} x$ \\
    & $N_1 = \alpha^{48}x^2 + \alpha^{33} x, V_1 = y$ & $g_1 = \alpha^{16} x^2 + \alpha x + \alpha^{31} y$ \\ \hline

\multirow{3}{*}{3} & $s_0 = \alpha^7, s_1 = 0$ & $\Delta_0 = \alpha^7, \Delta_1 = 0$ \\		
		& $N_0 = \alpha^{33}x^4 + \alpha^3 x^2, V_0 = y^2$ & $g_0 = x^8 + \alpha^{39} x^4 + \alpha^{34} x^2 + \alpha^{38} x$ \\
    & $N_1 = \alpha^{55} x^2 + \alpha^{40} x, V_1 = \alpha^7 y$ & $g_1 = \alpha^{16} x^2 + \alpha x + \alpha^{31} y$ \\ \hline

\multirow{3}{*}{4} & $s_0 = \alpha^{17}, s_1 = \alpha^{47}$ & $\Delta_0 = \alpha^{50}, \Delta_1 = \alpha^8$ \\		
		& $N_0 = \alpha^{47}x^4 + \alpha^4 x^2 + \alpha^{24} x, V_0 = \alpha^{14} y^2 + \alpha^{54} y$ & $g_0 = \alpha^8 x^8 + \alpha^{47} x^4 + \alpha^{46} x^2 + \alpha^{45} x + \alpha^{18} y$ \\
    & $N_1 = \alpha^{17} x^4 + \alpha^{37} x^2 + \alpha^{57} x, V_1 = \alpha^{47} y^2 + \alpha ^{24} y$ & $g_1 = \alpha^{32} x^4 + \alpha^{52} x^2 + \alpha^9 x + \alpha^{62} y^2 + \alpha^{39} y$ \\ \hline

\multirow{3}{*}{5} & $s_0 = \alpha^{31}, s_1 = \alpha$ & $\Delta_0 = \alpha^{18}, \Delta_1 = \alpha^{16}$ \\		
		& $N_0 = \alpha^{34}x^8 + \alpha^{37} x^4 + \alpha^{10} x^2 + \alpha^{58} x, V_0 = \alpha^{31} y^4 + \alpha^{25} y$ & $g_0 = \alpha^{24} x^8 + \alpha^{22} x^4 + \alpha^{47} x^2 + \alpha^{58} x + \alpha^{17} y^2 + \alpha^{49} y$ \\
    & $N_1 = 0, V_1 = 0$ & $g_1 = \alpha x^8 + \alpha^4 x^4 + \alpha^{40} x^2 + \alpha^{25} x + \alpha^{61} y^4 + \alpha^{55} y$ \\ \hline

\multirow{2}{*}{6} & $s_0 = \alpha^{16}, s_1 = 0$ & $\Delta_0 = \alpha^{6}, \Delta_1 = \alpha^{46}$ \\		
		& $N_1 = 0, V_1 = 0, N_0 = 0, V_0 = 0$ & $g_1 = \alpha^4 x^4 + x^2 + \alpha^{29} x + \alpha^{4} y^4 + y^2 + \alpha^{29} y$ \\ \hline
\end{tabular}
\end{center}
\end{table*}

\section{Decoding of KK Codes} \label{sec: generalKK}

For a KK code over GF$(q^m)$ as described in Section~\ref{sec: introKK}, the decoding algorithm in~\cite{kotter_it08} finds a \textbf{minimum} solution to~(\ref{equ: KKinterpolation}) based on an interpolation procedure. In this section, we will show that this list-1 decoding algorithm is a special case of our general interpolation algorithm over free $L[x]$-modules, where  $L[x]$ is the ring of linearized polynomials over GF$(q^m)$.

%

\begin{lemma} \label{lemma: KKspecialcase}
When $L= 1$, Algorithm~\ref{alg: GILinearized} reduces to the Sudan-style list-1 decoding algorithm in~\cite{kotter_it08}.
\end{lemma}
\begin{proof}
We assume that the condition of decodability~\cite{kotter_it08} is satisfied so that an interpolation approach works to gives a solution of $Q(x, y)$. Given the linearized polynomial ring $L[x]$ over GF$(q^m)$, we set $L=1$, choose a set $B = \{ b_0, b_1\} = \{ x, y\}$ as a basis, and construct the same free $L[x]$-module $V = \{Q(x,y) \}$ with the same ordering as that in Section~\ref{sec: GeneralGb}. Hence Algorithm~\ref{alg: GILinearized} has exactly the same initial values and the same update rules as the Sudan-style list-1 decoding algorithm in~\cite{kotter_it08} (it should be pointed out that the pseudocode in ~\cite{kotter_it08} contains a typo, and no update is going to take place when both discrepancies are zero~\cite{kschischang_privatecomm10}). Hence we only have to show that the final output of the two algorithms are the same (of the same order).

We consider the definition of minimum in Algorithm~\ref{alg: GILinearized} and the notion of $x$-minimal and $y$-minimal in~\cite{kotter_it08}. According to the definition in~\cite{kotter_it08}, $f_0^{(i)}(x,y)$ is $x$-minimal if it interpolates through the first $i$ points and is a minimal polynomial under $\prec$, while its leading term is in $x$. Comparing this definition to that in our general interpolation construction, we find that this $f(x,y)$ is a minimum in $T_{i,0}$, hence $f_0^{(i)}(x,y) =_o g_{i,0}$. Similarly, $f_1^{(i)}(x,y)$ being $y$-minimal means that $f_1^{(i)}(x,y) =_o g_{i,1}$ in Algorithm~\ref{alg: GILinearized}. Since KK's decoding algorithm finds $x$-minimal and $y$-minimal bivariate linearized polynomials in each step, it works the same as Algorithm~\ref{alg: GILinearized} during intermediate steps. Finally, KK's decoding algorithm outputs the one with a smaller $(1, k-1)$-weighted degree, which equals to finding the minimum among $g_{C,0}$ and $g_{C,1}$ as performed in Algorithm~\ref{alg: GILinearized}. Hence our general interpolation algorithm reduces to the list-1 decoding algorithm in~\cite{kotter_it08} when $L=1$.
%
\end{proof}

The proof of Lemma~\ref{lemma: KKspecialcase} also indicates that when used to decode KK codes,  Algorithm~\ref{alg: GILinearized} requires finite field multiplications of order $O(n(n+k))$ .

\section{List Decoding of MV Codes} \label{sec: generalMV}

In~\cite{mahdavifar_isit10}, the list decoding procedure first constructs a multivariate polynomial $Q(x, y_1, y_2,\ldots, y_L)$ that interpolates through a number of given points as indicated by (\ref{equ: lLinter}). Hence we call this process the interpolation step of the list decoding of MV codes. No specific algorithm is mentioned in~\cite{mahdavifar_isit10} on how to get this multivariate polynomial. Of course, a nonzero solution can be obtained by solving the corresponding homogeneous systems, but with high computational complexity. Here, we utilize the general interpolation over free $L[x]$-modules to solve this problem efficiently. The complexity of our algorithm is compared to that of  solving homogeneous equations.

As in Section~\ref{sec: GeneralGb}, we have to construct a free module for a given ring, and define relative operations so that Algorithm~\ref{alg: GILinearized} can be carried out. We consider an $l$-dimensional MV codes over GF$(q^{ml})$ defined in Section~\ref{sec: introMV}, with a message vector length of $k$ and dimension of subspace $l$. In this case, the linearized polynomials ring $L[x]$ is defined over GF$(q^{ml})$, and a set $B = \{ b_0, b_1,\ldots, b_L\} = \{ x, y_1,\ldots, y_L\}$ is selected to form a free $L[x]$-module $V = \{Q(x,y_1, \ldots, y_L)\}$. Following a similar definition of the multiplication between $L[x]$ and $V$, the module $V$ is constructed in the same way as in Section~\ref{sec: GeneralGb}. Hence an element $Q(x,y_1, \ldots, y_L) \in V$ can be written as $Q(x,y_1, \ldots, y_L)=Q_0(x)+Q_1(y_1)+\cdots+Q_L(y_L)$, called a \emph{multivariate linearized polynomial}, where $Q_i(x)\in L[x]$ for $i =0,1,\ldots, L$.  Following similar process as in the previous section, we can claim that $V$ is also a vector space over GF$(q^{ml})$ with a vector space basis $M =  \{ x^{[i]} \circ b_{j}, i\ge 0, j = 0, 1,\ldots, L\}$. Then we define an total ordering on $M$ as follows. We write $b_j^{[i]} < b_j^{[i+1]}$ for $j \in \{0,1,\ldots, L\}$ and $i \ge 0$, and write $b_j^{[i]} < b_{j'}^{[i']}$ if $j(k-1) + i = (j+1)(k-1) + i'$ and $j < j'$ for $j,j'\in \{0,1,\ldots, L\}$ and $i,i' \ge 0$. Then $M = \{\phi_j\}_{j \ge 0}$ such that $\phi_i < \phi_j$ when $i < j$. Hence we can define the leading monomial and the order of any $Q \in V$ in the same way as in Section~\ref{section: GeneralInterpolation}, as well as the minimum elements in a subset of $V$. Finally, a set of linear functionals $D_i$ for $i=1,2,\ldots, (t+l)m$ from $V$ to GF$(q^{ml})$ are also defined to be evaluations of multivariate linearized polynomials by the given points, as indicated in (\ref{equ: lLinter}). Here, the total number of points to be interpolated in (\ref{equ: lLinter}) is $(t+l)m$, hence the numbers of linear functionals $D_i$ and of the kernels $K_i$ are $(t+l)m$. Furthermore, the kernels $K_i$ are also $L[x]$-submodules by Lemma~\ref{lemma: LxKernels}. In summary, the interpolation problem in~(\ref{equ: lLinter}) is to find a nonzero $Q\in V$ such that $Q \in \overline{K}_{(t+l)m}$. Hence this is a general interpolation problem over free $L[x]$-modules, thus can be solved by Algorithm~\ref{alg: GILinearized}, which gives a \textbf{minimum} nonzero solution to (\ref{equ: lLinter}), as stated in the following lemma.
\begin{lemma} \label{lemma: 1Lalg1}
The general interpolation algorithm solves the interpolation problem of the list decoding algorithm for $l$-dimensional MV codes if the dimension of the error $t < lL - L(L+1) \frac{k-1}{2m}$.
\end{lemma}
\begin{proof}
As shown in~\cite{mahdavifar_isit10}, when $t < lL - L(L+1) \frac{k-1}{2m}$, there exist nonzero solutions for the interpolation step of the list decoding algorithm fir MV codes. Hence we assume $t < lL - L(L+1) \frac{k-1}{2m}$, then Algorithm~\ref{alg: GILinearized} solves the interpolation problem by finding a \textbf{minimum} nonzero solution to~(\ref{equ: lLinter}), when we adopt the free modules and related operations as described above.
\end{proof}

For Algorithm~\ref{alg: GILinearized}, we set $g_{0,0} = x, g_{0,i} = y_i$ for $i = 1, 2,\ldots, L$ in the initialization step. The update rules in the intermediate iteration steps are the same as in Section~\ref{sec: GeneralGb}, only that we have to use the new ordering related definitions in this section to determine a minimum among the $L+1$ candidates.

Finally we discuss the complexity of Algorithm~\ref{alg: GILinearized} when used to decode $l$-dimensional MV codes. As mentioned above, a nonzero multivariate linearized polynomial $Q(x, y_1, \ldots, y_L)$ can also be obtained by solving the homogeneous system determined by (\ref{equ: lLinter}). The size of the coefficient matrix is $(t+l)m \times (ml (L+1) - \frac{k-1}{2}L(L+1))$. Gaussian elimination has a complexity of  $O((t+l)^2 m^2 (ml (L+1) - \frac{k-1}{2}L(L+1)))$. Given the fact that $ml-L(k-1)-1 \ge 0$ (the $q$-degree of $Q_L(y_L)$ has to be nonnegative), this complexity is $O(L^2 m^2 (t+l)^2 (k-1))$. On the other hand, for Algorithm~\ref{alg: GILinearized}, we have $C = (t+l)m$ linear functionals in this case and a total of $L+1$ elements in the basis of the free module, and the highest $q$-degree of the linearized polynomials in $x$ is at most $ml-1$ among all the iterations. Since the linear functional operation and the multiplication between elements in the ring and elements in the module are defined in the same manner as in KK codes case, the complexity of Algorithm~\ref{alg: GILinearized} is of $O(L^2 m^2 (t+l)l)$. Hence the general interpolation approach is more efficient when compared to solving linear equations.

\section{Conclusion} \label{sec: conclusion}
In this paper, we investigate the general interpolation problem over free modules of a linearized polynomial ring, and propose a general interpolation algorithm. Our general interpolation algorithm are used to decode Gabidulin codes and KK codes. Comparisons are made between our algorithm for Gabidulin codes and Loidreau's decoding algorithm. Analysis shows that the Sudan-style list-1 decoding algorithm for KK codes is a special case of our general interpolation algorithm. Our general interpolation approach also applies to find the multivariate linearized polynomial in the list decoding of MV codes by Mahdavifar and Vardy, while currently no efficient algorithm is available to accomplish the task.
\bibliographystyle{IEEEtran}
\bibliography{interpolation0102}

\end{document}